\documentclass[submission,copyright,creativecommons]{eptcs}
\providecommand{\event}{SR 2013} 
\usepackage{breakurl}             
\usepackage{amssymb, amsmath}
\usepackage{tikz}
\usepackage{multicol}
\usetikzlibrary{positioning,arrows}

\title{Infinite games with uncertain moves}
\author{Nicholas Asher \ and \ Soumya Paul
\institute{IRIT, Universit\'e Paul Sabatier, 118 Route de Narbonne, 31062
  Toulouse, France.}
\email{\{nicholas.asher,soumya.paul\}@irit.fr}
\thanks{We thank ERC grant 269427 for research support.}}
\def\titlerunning{Infinite games with uncertain moves}
\def\authorrunning{N. Asher \& S. Paul}
\begin{document}
\newcommand{\win}{\mathit{Win}}
\newcommand{\X}{\mathcal{X}}
\newcommand{\occ}{\mathit{occ}}
\newcommand{\infy}{\mathit{inf}}
\newcommand{\reach}{\mathit{Reach}}
\newcommand{\safe}{\mathit{Safe}}
\newcommand{\muller}{\mathcal{F}}
\newcommand{\mul}{\mathit{Muller}}
\newcommand{\pref}{\mathit{pref}}
\newcommand{\nec}{\mathit{nec}}
\newcommand{\red}{\mathit{red}}
\newcommand{\mini}{\mathit{min}}
\newcommand{\maxi}{\mathit{max}}
\newcommand{\ext}{\mathit{ext}}
\newcommand{\tset}{\mathcal{T}}
\newcommand{\aset}{\mathcal{A}}
\newcommand{\tunion}{\bigsqcup}
\newcommand{\tint}{\sqcap}
\newcommand{\tu}{\mathcal{U}}
\newcommand{\uc}{\mathcal{U}}
\newtheorem{definition}{Definition}
\newtheorem{theorem}{Theorem}
\newtheorem{prop}{Proposition}
\newtheorem{cor}{Corollary}
\newtheorem{lemma}{Lemma}
\newtheorem{obs}{Observation}

\newenvironment{rem}[1]{\vspace{1ex}\noindent{\bf
Remarks}\hspace{0.5em}}{ }
\newcommand{\qed}{\hfill \mbox{\raggedright \rule{.07in}{.1in}}}

\newenvironment{proof}{\vspace{1ex}\noindent{\bf Proof}\hspace{0.5em}}
	{\hfill\qed\vspace{1ex}}
\newenvironment{proofof}[1]{\vspace{1ex}\noindent{\bf Proof of #1}\hspace{0.5em}}
	{\hfill\qed\vspace{1ex}}

\maketitle

\begin{abstract}
We study infinite two-player games where one of the players is unsure about
the set of moves available to the other player. In particular, the set
of moves of the other player is a strict superset of what she assumes
it to be. We explore what happens to sets in various levels of
the Borel hierarchy under such a situation. We show that the sets at
every alternate level of the hierarchy jump to the next higher level.
\end{abstract}

\section{Introduction}\label{secintro}
Infinte two-player games have attracted a lot of attention 
and found numerous applications in the
fields of topology, descriptive set-theory, computer science etc. Examples of such types of games
are: Banach-Mazur games,
Gale-Stewart games, Wadge games, Lipschitz games, 
etc. \cite{Kechris,Kanamori,Wadge,GS53}, and they each characterize
different concepts in descriptive set theory.

These games are typically played between two players, Player 0 and Player 1, who take turns
in choosing finite sequences of elements (possibly singletons) from a fixed set $A$ (finite or
infinite) which is called the alphabet. This process goes
on infinitely and hence defines an infinite sequence $u_0u_1u_2\ldots$
of finite strings which in itself is an infinite string over the set
$A$. In addition, the game has a winning condition $\win$ which is a
subset of the set of infinite strings over $A$, $A^\omega$. Player 0
is said to win the game if the sequence $u_0u_1u_2\ldots$ is in
$\win$. Player 1 wins otherwise.

In addition to their applications in descriptive set-theory and topology, such
games have also been used in computer science in the fields of
verification and synthesis of reactive systems \cite{LNCS2500}. 
The verification problem is modeled as a game between two players:
the system player and the environment player. The winning set $\win$
is specified using formulas in some logic, LTL, CTL, $\mu$-calculus
etc. The goal of the system player is to meet the specification along
every play and
that of the environment player is to exhibit a play which does not meet
it. To verify the system then amounts to show that the system player
has a winning strategy in the underlying game and to find this
strategy.

When $\win$ is specified using the usual logics, it corresponds to
sets in the low levels of the Borel hierarchy. It is known that the
complexity of the winning strategy increases with the increase in the
level of the Borel hierarchy to which $\win$ belongs \cite{Serre04}. For instance, in
Gale-Stewart games, {\sf reachability}, {\sf safety} and {\sf Muller}
are winning conditions in the $\Sigma_1^0, \Pi_1^0$ and $\Sigma_2^0$ levels of the Borel hierarchy
respectively and a player has positional winning strategies for
reachability and safety but needs memory to win for the Muller
condition. However it was shown in \cite{GH82,Mos91} that a finite
amount of memory suffices. The notion of Wadge reductions also
formalises this increase in complexity of the sets along the Borel hierarchy. 


Such games (esp. Banach-Mazur and Gale-Stewart games) also find
applications in linguistics.  \cite{AP12} shows that conversations have a
topological structure similar to that of Banach-Mazur games and explores how
the different types of objectives of 
conversations correspond to different levels in the Borel hierarchy
depending on their complexity. \cite{AP12} also applied of the classical
results from the literature of Banach-Mazur games to the
conversational setting. \cite{map12} applies Gale-Stewart games to the study of politeness.

In this paper, we look at what happens to sets in the Borel hierarchy
when the underlying alphabet is expanded. That is, the alphabet is
changed from $A$ to $B$ such that $B$ is a strict superset of $A$. We
show that sets at every alternate level of the Borel hierarchy undergo
a jump to the next higher level. More precisely, a set at level $n$ of
the hierarchy with alphabet $A$ moves to level $n+1$ when the alphabet
is expanded to $B$. This process goes on for all countable levels and
stabilises at $\omega$. 

Our result has consequences for both formal verification and
linguistic applications some of which we elucidate in the concluding section.

The rest of the paper is organised as follows. In Section
\ref{secprel} we formally introduce the necessary concepts and give
the required background for the paper. Then in Section \ref{secres} we
state and prove the main results of the paper. Finally we conclude
with some interesting consequences in Section \ref{seccon}.

\section{Preliminaries}\label{secprel}
In this section we present the necessary background required for the
paper. Although we define most of the concepts used in the paper, we
assume some familiarity with the basic notions of topology and
set-theory.
\subsection{Open and closed sets}
Let $A$ be a non-empty set. We sometimes
refer to $A$ as the {\sf alphabet}. For any subset $X$ of $A$, as
usual, we denote by $X^*$ the set of finite strings over $X$ and by
$X^\omega$, the set of countably infinite strings over $X$. For any
string $u \in A^*\cup A^\omega$ we denote the $i$th element of $u$ by $u(i)$. The set of {\sf prefixes} of $u$ are
all strings $v \in A^*$ such that $u = vv'$ for $v'\in A^*\cup A^\omega$.

We define a topology on $A^\omega$, the standard topology (also known
as the Cantor topology) on the set of infinite
strings over $A$. This topology can be defined in at least three
equivalent ways. The first way is to define the discrete topology on
$A$ and then assign $A^\omega$ the product topology. The second way is
to explicitly define the open sets of the topology. The open sets are
given by sets of the form $XA^\omega$ where $X$ is a subset of $A^*$.
Thus an open set is a set of finite strings over $X$ followed by their all
possible continuations. For a set $X\subseteq A^*$, we denote the open
set $XA^\omega$ by 
$O_A(X)$ or simply by $O(X)$ when the underlying alphabet $A$ is clear
from the context. When $X$ is a singleton $\{u\}$, we abuse notation
to denote the open set $uA^\omega$ by $O_A(u)$.
Example 1 illustrates these concepts.

\vspace{.1in}
\noindent
{\bf Example 1.} 
Let $A=\{a,b,c\}$. Then $abcA^\omega$ is an open set and so is
$abA^\omega \cup baA^\omega$. The complement of the set $abcA^\omega$
is the set $X$ of all strings that do not have $abc$ as their
prefix. This is a closed set.  
\vspace{.1in}

Yet another equivalent way to define the topology is to give an
explicit metric for it. Given two strings, $u_1,u_2 \in A^\omega$, the distance
between them $d(u_1,u_2)$ is defined to be $1/2^{n(u_1,u_2)}$, where
$n(u_1,u_2)$ is the first index where $u_1$ and $u_2$ differ from each
other. Thus the above topology is metrisable. Henceforth, when we use
the term `set' we shall mean a subset 
of $A^\omega$.

Note that the set $(\overline{abcA^\omega})$ in the above example is also
open. That is because it is a union of the open sets $O(aa),
O(ac), O(b)$ and $O(c)$. Such sets, which are both open and
closed are called {\sf clopen} sets. So what is a set which is open
but not closed (and vice versa)?

\begin{prop}[\cite{PP}]
  If $A$ is a finite alphabet, a subset of $A^\omega$ is clopen if and
  only if it is of the form $XA^\omega$ where $X$ is a finite subset
  of $A^*$.
\end{prop}

Thus if $A$ is finite then a set of the form $XA^\omega$ where $X$ is
an infinite subset of $A^*$ is open but not closed. If $A$ is infinite, the subsets of $A^\omega$ of the form
$XA^\omega$, where $X$ is a set of words of bounded length of $A^*$
are clopen. However there might exist clopen sets which are not of this form.


\subsection{The Borel hierarchy}
A set of subsets of $A^\omega$ is called a {\sf $\sigma$-algebra} if
it is closed under countable unions and complements. Given a set $X$,
the smallest $\sigma$-algebra containing $X$ is called the
$\sigma$-algebra {\sf generated by $X$}. It is equivalent to the
intersection of all the $\sigma$-algebras containing $X$. The sigma
algebra generated by the open sets of a topological space is called
the {\sf Borel} $\sigma$-algebra and its sets are called the {\sf
  Borel sets}.

The Borel sets can also be defined inductively. This gives a natural
hierarchy of classes $\Sigma_\alpha^0$ and $\Pi_\alpha^0$ for $1\leq
\alpha <\omega_1$. Let $\Sigma_1^0$ be the set of all open
sets. $\Pi_1 = \overline{\Sigma_1^0}$ is the set of all closed
sets. Then for any $\alpha>1$ where $\alpha$ is a successor ordinal,
define $\Sigma_\alpha^0$ to be the countable union of all
$\Pi_{\alpha-1}^0$ sets and define $\Pi_\alpha^0$ to be the complement
of $\Sigma_\alpha^0$. For a limit ordinal $\eta,\ 1<\eta<\omega_1$,
$\Sigma^0_\eta$ is defined as $\Sigma_\eta^0 = \bigcup_{\alpha <
  \eta}\Sigma_\alpha^0$ and $\Pi_\eta^0 =
\overline{\Sigma_\eta^0}$. The infinite hierarchy thus generated is
called the {\sf Borel hierarchy} and they together form the Borel
algebra. It is known \cite{PP} that if the space is metrisable and the
underlying alphabet contains at least two elements, then the hierarchy
is indeed infinite, that is,
the containments, $\Sigma_\alpha^0 \subset \Sigma_{\alpha+1}^0$ and $\Pi_\alpha^0
\subset \Pi_{\alpha+1}^0$ are strict.


\subsection{Wadge reductions and complete sets}
Let $A$ and $B$ be two alphabets. A function $f: A^\omega \rightarrow
B^\omega$ is said to be continuous if for every open subset $Y\subseteq
B^\omega$, $f^{-1}(Y)$ is also open.

A set $X\subseteq A^\omega$ is said
to {\sf Wadge reduce} to another set $Y\subseteq B^\omega$, denoted $X\leq_W Y$, if there
exists a continuous function $f: A^\omega \rightarrow B^\omega$ such
that $f^{-1}(Y) = X$. 

Let $A$ be an alphabet. A set $X \subseteq A^\omega$ is said to be
$\Sigma_\alpha^0$ (resp. $\Pi_\alpha^0$) {\sf complete} if $X\in
\Sigma_\alpha^0$ (resp. $X\in \Pi_\alpha^0$) and for any other
alphabet $B$ and for any $\Sigma_\alpha^0$ (resp. $\Pi_\alpha^0$) set
$Y\subseteq B^\omega$, $Y\leq_W X$. Intuitively, given a class of sets
$\Gamma$, the complete sets of that class represent the sets which are
structurally the most complex in that class.

For the Borel hierarchy, completeness can be characterised in the following simple way:

\begin{prop}[\cite{PP}]\label{propcomplete}
Let $X\subseteq A^\omega$. Then $X$ is $\Pi_\alpha^0$ (resp. $\Sigma_\alpha^0$) complete if and only if $X\in \Pi_\alpha^0\setminus\Sigma_\alpha^0$ (resp. $\Sigma_\alpha^0\setminus\Pi_{\alpha-1}^0$).
\end{prop}

\subsection{Infinite games}
Let $A$ be an alphabet. An infinite game on $A$ is played between two players, Player 0 and Player 1, who take turns
in choosing finite sequences of elements (possibly singletons) from a fixed set $A$ (finite or
infinite) which is called the alphabet. This process goes
on infinitely and hence defines an infinite sequence $u_0u_1u_2\ldots$
of finite strings which in itself is an infinite string over the set
$A$. In addition, the game has a winning condition $\win$ which is a
subset of the set of infinite strings over $A$, $A^\omega$. Player 0
is said to win the game if the sequence $u_0u_1u_2\ldots$ is in
$\win$. Player 1 wins otherwise.

In a Banach-Mazur game, each player at her turn chooses a finite non-empty sequence of elements from $A$ while in a Gale-Stewart game the players are restricted to choosing just single elements from $A$. An infinite game can also be imagined to be played on a graph $G=(V,E)$ where the set of vertices $V$ is partitioned into $V_0$ and $V_1$ which represent the Player 0 and 1 vertices respectively. The game starts at an initial vertex $v_0\in V$ and the players take turns in moving a token along the edges of the graph depending on whose vertex it is currently. This process is continued ad infinitum and thus generates an infinite path $p$ in the graph $G$. Player 0 wins if and only if $p\in \win$ where $\win$ is a pre-specified set of infinite paths.

\section{Results}\label{secres}
In this section we present the main results of this paper. Given a
subset $B$ of an alphabet $A$ the topology of $B^\omega$ where the
open sets are given by $O\cap B^\omega$ for every open set $O$ of
$A^\omega$ is called the relative topology of $B^\omega$ with respect
to $A^\omega$. However we are interested in the opposite
question. What happens when the alphabet expands? In
particular, we show that when the alphabet set changes from $A$ to $B$
(say) such that $B$ is a strict superset of $A$ then the sets in the alternative
levels of the Borel hierarchy undergo a jump in levels.

\begin{lemma}\label{lembase}
  Let $A$ and $B$ be two alphabets such that $A \subsetneq B$. An open
  set $O$ in the space $A^\omega$ jumps to $\Sigma_2^0$ in the space
  $B^\omega$. A closed set $C$ in the space $A^\omega$ remains closed
  in $B^\omega$.
\end{lemma}

\begin{proof}
The proof is by carried out by coding the open set $O$ in the space
$B^\omega$ and demonstrating a complete set for $B^\omega$.

  Let $O$ be an open set in $A^\omega$. Then $O$ is of the form
  $XA^\omega$ where $X\subseteq A^*$. Let
  $\X_\beta$ be an indexing of the set $X$.

  Each element $u$ of $X$ gives the open set $O_A(u)$ which is a
  subset of $A^\omega$.  Now, when we move to the alphabet $B$, the
  set $O_B(u)$ is the set of strings which have $u$ as a prefix and all possible
  continuations using letters of $B$. Thus $O_B(u)$ is a strict
  superset of $O_A(u)$. Hence, we need to restrict 
  $O_B(u)$ in $B^\omega$ such that we obtain a set which
  is equal to $O_A(u)$ in $A^\omega$. One way to do do so
  is as follows. Consider
  all the finite continuations of $u$ in letters from $A$. Let
  $\uc_\gamma$ be an indexed set of all these continuations. Then
  $O_A(u)$ is the set 
  \begin{equation}\label{eqrep}
    O_A(u) = \bigcap O_B(u'),\ u'\in \uc_\gamma
  \end{equation}
  which is a closed set, being an arbitrary intersection of closed sets.

  Thus the set $O$ can be represented in $B^\omega$ as
  $$O= \bigcup O_A(u),\ u\in \X_\beta$$
  each of which by (\ref{eqrep}) is a closed set. Hence $O\in
  \Sigma_2^0$ in the space $B^\omega$.

  Next we demonstrate a $\Sigma^0_1$ set $O$ in a space
  $A^\omega$ which is complete for $\Sigma^0_2$ in a space $B^\omega$
  where $A\subsetneq B$. Let $A=\{a,b\}$ and $B=\{a,b,c\}$. Let
  $X=\{ab, abab, ababab,\ldots\} \subset A^*$ and let
  $O=XA^\omega$. Then $O$ is open. Each subset $O_A(u),\ u\in X$ is represented in $B^\omega$
as
$$O_A(u) = O_B(u) \cap O_B(ua) \cap O_B(ub) \cap O_B(uaa) \cap
O_B(uab) \cap
O_B(uba) \cap O_B(ubb) \cap \ldots$$
and
$$O= O_A(u_1) \cup O_A(u_2) \cup \ldots,\ u_i\in X$$
Hence $O$ is a $\Sigma_2^0$ set in $B^\omega$.

To show that $O$ is $\Sigma_2^0$ complete for $B^\omega$ we use Proposition \ref{propcomplete}. $O$ is not open in $B^\omega$. Indeed, because otherwise, there exists a finite string $u$ whose all
possible continuations with letters from $B$ are in $O$ and that is a
contradiction. $O$ is also not closed in $B^\omega$. To see this, note
that the complement of $O$, $\overline{O}$ in $A^\omega$ is the set $XA^\omega$ where $X\subseteq A^*$ is given as $X = \{b, aa, abb, abaa,
\ldots\}$. For $O$ to be closed in $B^\omega$, $\overline{O}$ should
be open in $B^\omega$. This means that there should exist a finite string $v$ whose all
possible continuations with letters from $B$ are in $\overline{O}$ which is again a contradiction.

Thus $O\notin \Sigma_1^0$ and $O\notin \Pi_1^0$ in $B^\omega$ and
hence it is complete for $\Sigma_2^0$ in $B^\omega$.



Next suppose $C$ is a closed set in $A^\omega$. We show how to
represent $C$ in $B^\omega$. Let
  $\tu_\beta$ be the indexed set of prefixes of $C$. Then $C$ can
be represented in $B^\omega$ as
$$C= \bigcap O_B(v),\ v\in \tu_\beta$$
Each $O_B(v)$ is a closed set in $B^\omega$ and hence $C$ being an
arbitrary intersection of closed sets in $B^\omega$ is closed. Thus
$C\in \Pi_1^0$ in $A^\omega$ remains $\Pi_1^0$ in $B^\omega$.
\end{proof}

We generalise the above Lemma to the entire Borel hierarchy in the following theorem.

\begin{theorem}\label{thmjump}
  Let $A$ and $B$ be two alphabets such that $A \subsetneq B$. We have
  the following in the Borel hierarchy:
  \begin{enumerate}
  \item For $1\leq \alpha < \omega$ and $\alpha$ odd, 
    \begin{enumerate}
    \item a set $X\in
      \Sigma_\alpha^0$ in the space $A^\omega$ jumps to
      $\Sigma_{\alpha+1}^0$ in the space $B^\omega$
    \item a set $X\in \Pi_\alpha^0$ in the space $A^\omega$ remains
      $\Pi_\alpha^0$ in the space $B^\omega$.
    \end{enumerate}
  \item For $1\leq \alpha < \omega$ and $\alpha$ even, 
    \begin{enumerate}
    \item a set $X\in
      \Sigma_\alpha^0$ in the space $A^\omega$ remains
      $\Sigma_{\alpha}^0$ in the space $B^\omega$
    \item a set $X\in \Pi_\alpha^0$ in the space $A^\omega$ jumps to
      $\Pi_{\alpha+1}^0$ in the space $B^\omega$.
    \end{enumerate}
\item For $\alpha \geq \omega$, a $\Sigma_\alpha^0$
  (resp. $\Pi_\alpha^0$) set remains $\Sigma_\alpha^0$
  (resp. $\Pi_\alpha^0$) on going from the space $A^\omega$ to
  $B^\omega$. That is, the sets stabilise.
\end{enumerate}
\end{theorem}

\begin{proof}
The proof is by induction on $\alpha$. For the base case, $\alpha=1$, the result follows from Lemma \ref{lembase}.

  The inductive case is relatively
  straightforward, given the inductive structure of the Borel
  hierarchy. For convenience, we subscript the sets with $A$ or $B$ to
  denote whether they are sets in $A^\omega$ or $B^\omega$
  respectively. 

  Suppose $1< \alpha < \omega$ and $\alpha$ is
  odd. Then
  \begin{align*}
    \Sigma_{\alpha,X}^0 = &\bigcup \Pi^0_{\alpha-1,X} \text{ [by
        definition]}\\
    = &\bigcup \Pi^0_{\alpha,Y} \text{ [by induction hypothesis]}\\
    = &\Sigma_{\alpha+1,Y}^0
  \end{align*}
  
  \begin{align*}
    \Pi^0_{\alpha,X} = & \overline{\Sigma}^0_{\alpha,X}=
    \overline{\bigcup \Pi^0_{\alpha-1,X}} = \bigcap
    \overline{\Pi}^0_{\alpha-1,X} = \bigcap \Sigma^0_{\alpha-1,X} \text{ [by
        definition]}\\
    = & \bigcap \Sigma^0_{\alpha-1,Y} \text{ [by induction hypothetis]}\\
    = & \Pi^0_{\alpha,Y}
  \end{align*}
  
  Now, suppose $1< \alpha < \omega$ and $\alpha$ is
  even. Then
  \begin{align*}
    \Sigma_{\alpha,X}^0 = &\bigcup \Pi^0_{\alpha-1,X} \text{ [by
        definition]}\\
    = &\bigcup \Pi^0_{\alpha-1,Y} \text{ [by induction hypothesis]}\\
    = &\Sigma_{\alpha,Y}^0
  \end{align*}
  
  \begin{align*}
    \Pi^0_{\alpha,X} = & \overline{\Sigma}^0_{\alpha,X}=
    \overline{\bigcup \Pi^0_{\alpha-1,X}} = \bigcap
    \overline{\Pi}^0_{\alpha-1,X} = \bigcap \Sigma^0_{\alpha-1,X} \text{ [by
        definition]}\\
    = & \bigcap \Sigma^0_{\alpha,Y} \text{ [by induction hypothetis]}\\
    = & \Pi^0_{\alpha+1,Y}
  \end{align*}
  
  Finally, 
  $$\Sigma_{\omega,X}^0 = \bigcup_{n<\omega} \Sigma_{n,X}^0
  =\bigcup_{n<\omega} \Sigma_{n,Y}^0 =\Sigma_{\omega,Y}^0$$
  and
  $$\Pi_{\omega,Y}^0 = \overline{\Sigma}^0_{\omega,Y} =
  \Pi^0_{\omega,X}$$
\end{proof}

The above result can be concisely summarised by Figure \ref{fig:res}.

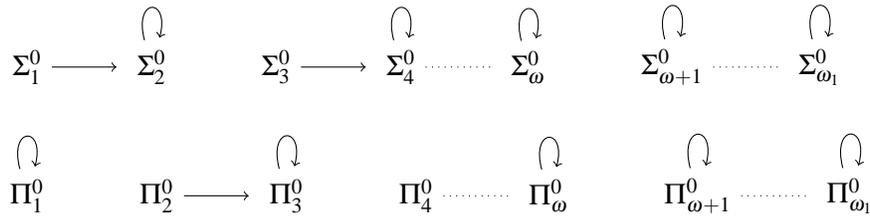
\begin{figure}[h]
  \vspace{0.5cm}
  \begin{center}
  \begin{tikzpicture}[shorten >=4pt,auto]
    \node (s1) {$\Sigma_1^0$};
    \node (s2) [right=of s1] {$\Sigma_2^0$};
    \node (s3) [right=of s2] {$\Sigma_3^0$};
    \node (s4) [right=of s3] {$\Sigma_4^0$};
    \node (so) [right=of s4] {$\Sigma_\omega^0$};
    \node (so1) [right=of so] {$\Sigma_{\omega+1}^0$};
    \node (soo) [right=of so1] {$\Sigma_{\omega_1}^0$};
    
    \node (p1) [below=of s1] {$\Pi_1^0$};
    \node (p2) [right=of p1] {$\Pi_2^0$};
    \node (p3) [right=of p2] {$\Pi_3^0$};
    \node (p4) [right=of p3] {$\Pi_4^0$};
    \node (po) [right=of p4] {$\Pi_\omega^0$};
    \node (po1) [right=of po] {$\Pi_{\omega+1}^0$};
    \node (poo) [right=of po1] {$\Pi_{\omega_1}^0$};
    
    \draw[->] (s1)--(s2);
    \draw[->] (s3)--(s4);
    \draw[->] (p2)--(p3);
    \draw[dotted] (s4)--(so);
    \draw[dotted] (p4)--(po);
    \draw[dotted] (so1)--(soo);
    \draw[dotted] (po1)--(poo);
    
    \path (s2) edge[loop above] (s2);
    \path (s4) edge[loop above] (s4);
    \path (so) edge[loop above] (so);
    \path (so1) edge[loop above] (so1);
    \path (soo) edge[loop above] (soo);
    
    \path (p1) edge[loop above] (p1);
    \path (p3) edge[loop above] (p3);
    \path (po) edge[loop above] (po);
    \path (po1) edge[loop above] (po1);
    \path (poo) edge[loop above] (poo);
  \end{tikzpicture}
  \caption{Jumps in the Borel hierarchy}
  \label{fig:res}
\end{center}
\end{figure}

\section{Applications}\label{seccon}
The result we showed has interesting consequences in the fields of
both formal verification and linguistics.
\subsection{Formal verification}
As we mentioned in the introduction, to formally verify a reactive
system $M$ (a piece of hardware or software which interacts with
users/environment), we often model the system as a finite graph
$G(M)$. Two players, the system player and the environment player then
play an infinite game on $G(M)$. 
The goal of the system player is to meet a certain specification on
all plays on $G(M)$ and
that of the environment player is to exibit a play which does not meet
it.

The result stated in this paper represents situations where the
system player is unsure about the exact moves of the environment
player. This shows that in such a situation, the system player might
have to strategise at a higher level of the hierarchy in order to
account for this uncertainty.

It can also be used to represent situations where the underlying model
might change (expand). Let $M$ be the original system and $M'$ be the
expanded system (which is generated from $M$ by the addition of a module say). If the objective of the system player in $G(M)$ was to
reach one of the states in some subset $R$ of $G(M)$ (reachability) then it is
enough for her to play positionally. However, in the bigger graph $G(M')$
she not only has to reach $R$ but also has to stay within the states
of the original graph $G(M)$ in order to achieve the same objective. This
is the Muller objective which is a level higher.

\vspace{.1in}
\noindent
{\bf Example 2.}
Consider the example shown in Figure \ref{fig:reachmul}. Player 0 nodes have been depicted as $\bigcirc$ and Player 1 nodes as $\Box$. Suppose initially the system is $M$ and the objective of Player 1 in $G(M)$ is to reach $v_3$. Then the winning set is the set of all sequences in $V=\{v_0,v_1,v_2,v_3\}$ in which $v_3$ occurs in some position. That is, $\win=\{u\ |\ \exists i,\ u(i)=v_3\}$. This is a reachability condition where the reachability set $R=\{v_3\}$. To win, Player 0 can either play $v_1$ or $v_2$ from $v_0$ and hence both these strategies are winning strategies for her. Now suppose the system expands to $M'$ where, in $G(M')$, it is possible for Player 1 to go to the new node $v_4$ from $v_1$. Also suppose $\win$ remains the same. Then $\win$ is no longer a reachability condition because then it would also include sequences involving the vertex $v_4$. It is rather a Muller condition where the Muller set $\muller =\{\{v_0,v_1,v_2,v_3\}\}$. However, note that Player 0 does not have a winning strategy in this game. That is because to win, she has to visit vertex $v_1$ infinitely often from which Player 1 can force the play through $v_4$ infinitely often.
\vspace{.1in}

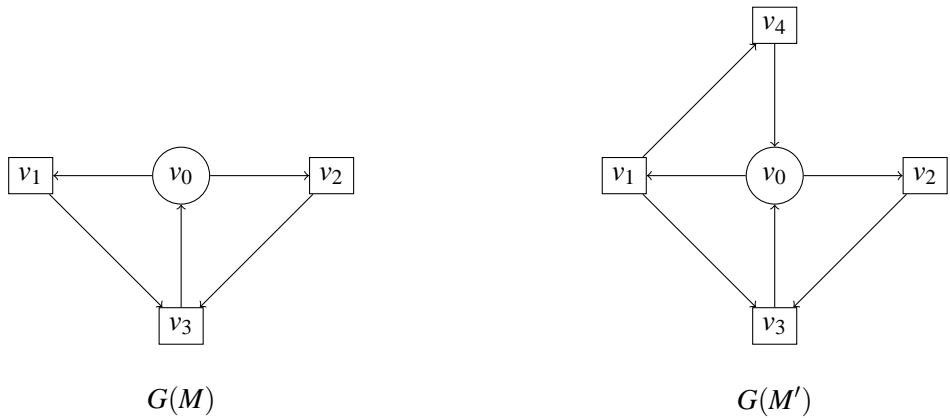
\begin{figure}[ht]
  \begin{minipage}[b]{0.45\linewidth}
    \centering
    \begin{tikzpicture}
      \node (v0) at (0,0) [circle,draw] {$v_0$};
      \node (v1) at (-2,0) [rectangle,draw] {$v_1$};
      \node (v2) at (2,0) [rectangle,draw] {$v_2$};
      \node (v3) at (0,-2) [rectangle,draw] {$v_3$};
      \node at (0,-3){$G(M)$};
      
      \draw[->] (v0) to (v1);
      \draw[->] (v0) to (v2);
      \draw[->] (v3) to (v0);
      \draw[->] (v1) to  (v3);
      \draw[->] (v2) to  (v3);
    \end{tikzpicture}
  \end{minipage}
\hspace{0.5cm}
\begin{minipage}[b]{0.45\linewidth}
\centering
\begin{tikzpicture}
\node (v0) at (0,0) [circle,draw] {$v_0$};
\node (v1) at (-2,0) [rectangle,draw] {$v_1$};
\node (v2) at (2,0) [rectangle,draw] {$v_2$};
\node (v3) at (0,-2) [rectangle,draw] {$v_3$};
\node (v4) at (0,2) [rectangle,draw] {$v_4$};
\node at (0,-3){$G(M')$};

\draw[->] (v0) to (v1);
\draw[->] (v0) to (v2);
\draw[->] (v3) to (v0);
\draw[->] (v4) to (v0);
\draw[->] (v1) to (v3);
\draw[->] (v2) to (v3);
\draw[->] (v1) to (v4);
\end{tikzpicture}
\end{minipage}
\caption{Jump from reachability to Muller}
\label{fig:reachmul}
\end{figure}

\subsection{Linguistics}
In \cite{AP12} we demonstrated what seems to be a compelling
similarity between human conversations and Banach-Mazur games. We showed how various conversational objectives
correspond to various levels of the Borel hierarchy and how strategies
of increasing complexity are called for to attain such
objectives. Our result shows that when Player 1 is unsure about what
Player 2 might say, it might be wise for her to strategise at a higher
level to account for this uncertainty. She
engages in a conversation, believing she is equipped with a strategy
for all the situations the other player might put her into when
suddenly the other player says something and she is left
dumbfounded.  

An example which still sticks in the memory of one of the authors after almost 20 years is
the memorable line by Senator Lloyd Bentsen in his Vice-Presidential debate with Dan Quale in 1984.  Quayle's strategy
in the debate was to counter the perception that he was too inexperienced to have the job, and he did this by drawing similarities
between his political career and former President John Kennedy's.   Quayle seemed to be doing a good job in achieving his objective or winning condition, when
Bentsen interrupted and said:
\begin{quote}
Sir, I knew Jack Kennedy.  I knew Jack Kennedy.  And you, sir, are no Jack Kennedy.
\end{quote}
Quayle's strategy at that point fell apart.  He had no effective come back and by all accounts lost the debate handily. 

The way we model this as follows.   Building on \cite{AP12}, we take each move in a game to be a discourse which may be composed of several, even
many clauses.  Abstractly, we consider such discourses as sequences of basic moves, which we will be the alphabet. In a situation of incomplete information about
the discourse moves,  the set of moves (or the alphabet) of the Banach
Mazur game being played by the players is different for the two
players. Player 0 has an alphabet $A$ (say) while Player 1 has an
alphabet $B$ such that $A\subsetneq B$. Player 0 may or may not be
aware of this fact.    

Thus, from the point of view of Player 0, if she is playing a
Banach-Mazur game where she is unsure of the set of moves available to
Player 1, it is better for her to strategise in such a way so as to
account for this jump in the winning set. In other words, if Player
0's winning condition is at a level $n$ (say) of the hierarchy, she is
better off strategising for level $n+1$ given that she is unsure of
Player 1's moves and given that a set at level $n$ might undergo a
jump to level $n+1$. Thus Quayle might have even won the debate had he strategiesed at a higher level expecting the unexpected.

\bibliographystyle{eptcs}
\bibliography{References}


\end{document}